\documentclass[10pt,conference]{IEEEtran}
\input{psfig.sty}
\input{epsf.sty}
\usepackage{hyperref} 
\usepackage{fancybox}   
\usepackage{psfig}      
\usepackage{graphicx}
\usepackage{amsmath}
\usepackage{color}
\usepackage{epsfig}
\usepackage{amssymb}
\usepackage{verbatim}
\usepackage{enumitem}
\usepackage{bbm}
\usepackage{cite}
\usepackage[]{algorithm2e}

\usepackage{caption}

\IEEEoverridecommandlockouts

\newcommand{\remove}[1]{}

\newtheorem{theorem}{Theorem}

\newtheorem{lemma}{Lemma}

\newtheorem{remark}{Remark}

\newtheorem{assumption}{Assumption}

\usepackage{setspace}
\usepackage{subfigure}

\begin{document}

 \title{Optimal Sensing and Data Estimation in a Large Sensor Network
 \thanks{Arpan Chattopadhyay and  Urbashi Mitra are with Ming Hsieh Department of Electrical Engineering, University of 
 Southern California, Los Angeles, USA. Email: \{achattop,ubli\}@usc.edu}\\
 \thanks{This work was funded by the following grants: ONR N00014-15-1-2550, NSF CNS-1213128, NSF CCF-1410009, AFOSR
FA9550-12-1-0215, and NSF CPS-1446901, and also by the UK Leverhulme Trust, the UK Royal Society of Engineers and the Fulbright Foundation.}
}

\author{
Arpan~Chattopadhyay, Urbashi~Mitra \\
}

\maketitle
\thispagestyle{empty}

\begin{abstract}
An energy efficient use of large scale sensor networks necessitates activating a subset of possible sensors for estimation at a fusion center. The problem is inherently combinatorial; to this end, a set of  iterative, randomized  algorithms are developed for sensor subset selection by exploiting the underlying statistics. Gibbs sampling-based methods are designed to optimize the estimation error and the mean number of activated sensors.  {\em The optimality of the proposed strategy is proven}, along with guarantees on their convergence speeds. Also, another new algorithm  exploiting stochastic approximation in conjunction with Gibbs sampling is derived for a constrained version of the sensor selection problem.  The methodology is extended to the scenario where the fusion center has access to only a parametric form of the joint statistics, but not the true underlying distribution.  Therein, expectation-maximization is effectively employed to  learn the distribution.  Strategies for iid time-varying data are also outlined.  Numerical results show that the proposed methods converge very fast to the respective {\em optimal} solutions, and therefore can be employed for optimal sensor subset selection  in practical sensor networks.
\end{abstract}


\section{Introduction}\label{section:introduction}
A wireless sensor network typically consists of a number of sensor  nodes deployed to monitor some physical process.  The sensor data is often delivered to a fusion center via wireless links. The fusion center,  based on the gathered data from the sensors, infers the state of the physical process and makes control decisions if necessary. 

Sensor networks have  widespread applications in various domains such as environmental monitoring, industrial process monitoring and control, localization, tracking of mobile objects, system parameter estimation, and even in disaster management. However, severe resource constraints in such networks necessitates  careful design and control strategies in order to attain a reasonable compromise between resource usage and network performance. One major restriction is that the nodes are battery constrained, which  limits the lifetime of the network. Also, low capacity of wireless channels due to transmit power constraint, heavy interference and unreliable link behaviour restricts the amount of data that can be sent to the fusion center per unit time. In some special cases, such as mobile crowdsensing applications (see \cite{schnitzler-etal15sensor-selection-crowdsensing}), a certain cost might be necessary in order to engage a sensor owned by a third party. All these constraints lead us to the fundamental question: how to select a small subset of sensors so that the observations made by these sensors are most informative for the effiicient inference of the state of the physical process under measurement?

Recent results have focused on optimal sequential sensor subset selection in order to monitor a random process modeled as Markov chain or linear dynamical system; see e.g. \cite{daphney-etal14active-classification-pomdp, daphney-etal13energy-efficient-sensor-selection, krishnamurthy07structured-threshold-policies, wu-arapostathis08optimal-sensor-querying, gupta-etal06stochastic-sensor-selection-algorithm, bertrand-moonen10sensor-selection-linear-mmse}. Sensor subset selection using these control-theoretic resullts are typically computationally expensive, and the low-complexity approximation schemes proposed in some of these papers (such as \cite{daphney-etal13energy-efficient-sensor-selection} and \cite{bertrand-moonen10sensor-selection-linear-mmse}) are not optimal. On the other hand, there appears to be limited work on optimal subset selection when sensor data is static and its distribution is known either absolutely or in parametric form; the major challenge in this problem is computational (\cite{wang-etal16efficient-observation-selection}), where the computational burden arises for  two reasons: (i) finding the optimal subset requires a search operation over all possible subsets of sensors, thereby requiring exponentially many number of computations, and (ii) for each subset of active sensors, computing the estimation error conditioned on the observation made by active sensors requires exponentially many computations. In \cite{wang-etal16efficient-observation-selection},  the problem of minimizing the minimum mean squared error (MMSE) of a vector signal using samples collected from a given number of sensors chosen  by the network operator is considered; a tractable lower bound to the MMSE is employed in a certain greedy algorithm to obviate the complexity in MMSE computation and  the combinatorial problem of searching over all possible subsets of sensors. In contrast, our paper deals with a general error metric (which could potentially be the MMSE or even the lower bound to MMSE as in  \cite{wang-etal16efficient-observation-selection}), and proposes Gibbs sampling based techniques for the optimal subset selection problem, in order to minimize a linear combination of the estimation error and the expected number of activated sensors.  To the best of our knowledge, ours is the first paper to use Gibbs sampling for {\em optimal} sensor subset selection with low complexity in the context of active sensing.  We also provide an algorithm based on Gibbs sampling and stochastic approximation, which is provably optimal  and which minimizes  the expected estimation error subject to a constraint on the expected number of activated sensors; this    technique  can be employed to solve many other constrained combinatorial optimization problems.\footnote{In this connection, we would like to mention that Gibbs sampling based algorithms were used in  wireless caching  \cite{chattopadhyay-etal16gibbsian-caching-arxiv}, but to solve an unconstrained problem. In the current paper, we  combine Gibbs sampling and stochastic approximation to solve a constrained optimization problem; this technique is   general and can iteratively solve many other constrained combinatorial optimization problems {\em optimally} with very small computation per iteration, while the approximation algorithms are not guaranteed to achieve optimality.}

\subsection{Organization and our contribution}\label{subsection:organization}
The  paper is organized as follows. 
 The system model is described in Section~\ref{section:system-model}.
 In Section~\ref{section:gibbs-sampling-unconstrained-problem}, we propose  Gibbs sampling based algorithms to minimize a linear combination of data estimation error and the number of active sensors. We prove convergence of these algorithms, and also provide a bound on  the convergence speed of one algorithm. 
 Section~\ref{section:gibbs-stochastic-approximation-unconstrained-problem} provides algorithm for  minimizing the estimation  error subject to a constraint on the mean number of active sensors. We propose a novel algorithm based on Gibbs sampling and stochastic approximation, and prove its convergence to the desired solution. To the best of our knowledge, this is a novel technique that can  be used for other constrained combinatorial optimization problems as well. We also discuss how the Gibbs sampling algorithm can be used when we have a hard constraint on the number of activated sensors. 
 Section~\ref{section:EM-for-parametric-distribution} discusses expectation maximization (EM) based algorithms when data comes from a parameterized distribution with unknown parameters.
  Numerical results on computational gain and performance improvement by using some of the proposed algorithms are presented in Section~\ref{section:numerical-results}. 
  Finally, we conclude in Section~\ref{section:conclusion}. All proofs are provided in the appendices.

We have also discussed in various sections how the proposed algorithms with minor modifications can be used for data varying in time in an i.i.d. fashion.

\section{System Model and Notation}\label{section:system-model}
\subsection{The network and data model}\label{subsection:network-data-model}
We consider a large connected single or multi-hop sensor network, whose sensor nodes are denoted by the set   $\mathcal{N}=\{1,2,\cdots,N\}$. 
Each node~$k \in \mathcal{N}$ is associated with a (possibly vector-valued) data $\underline{X}_k$, and we denote by $\underline{X}=\{\underline{X}_k\}_{k \in \mathcal{N}}$ the set of data which has to be reconstructed.  
A {\em fusion center}  determines the set of activated sensors, and estimates the data in each node given the limited observations only from the activated sensors.

While our methods assume static data from the sensors; these methods can be employed with good performance for data that varies in an iid fashion with respect to time.

\subsection{Reconstruction of sensor data}\label{subsection:reconstruction-problem}
We  denote the {\em activation state} of a sensor by $1$ if it is active, and by $0$ otherwise. 
We call $\mathcal{B}:=\{0,1\}^N$ the set of all possible configurations in the network, and denote a generic  configuration by $B$. Specifying a configuration is equivalent to selecting a subset $\mathcal{S}$ of active sensors. 
We denote by $B_{-j} \in \{0,1\}^{N-1}$ the configuration $B$ with its $j$-th entry removed.

The estimate of $\underline{X}$   at the fusion center is denoted by $\hat{\underline{X}}$.  The corresponding expected error under configuration $B \in \mathcal{B}$ is denoted by $\mathbb{E} d_B(\underline{X}, \hat{\underline{X}})= \sum_{k=1}^N \mathbb{E} d_B(\underline{X}_k,\hat{\underline{X}_k})$. Specifically, the  mean squared error (MSE) yields   $\mathbb{E} d_B(\underline{X}, \hat{\underline{X}})= \mathbb{E}(||\underline{X}-\hat{\underline{X}}||^2)=\sum_{k=1}^N \mathbb{E}(||\underline{X}_k-\hat{\underline{X}_k}||^2)$.  Let us denote the cost of activating a sensor by $\lambda$. Heterogeneous sensor classes with different priorities or weights can be straightforwardly accommodated and thus are not presented herein.

\subsubsection{The unconstrained problem}\label{subsubsection:unconstrained-optimization-problem-static-data}
Given a configuration $B \in \mathcal{B}$, the associated network cost is given by:
\small
\begin{equation}\label{eqn:cost-definition}
h(B):=\mathbb{E} d_B(\underline{X},\hat{\underline{X}})+\lambda ||B||_1 
\end{equation}
\normalsize

\noindent In the context of stastistical physics, one can view $h(B)$ as the potential under configuration $B$.  Our goal herein is to solve the following optimization problem:
\begin{equation}
\min_{B \in \mathcal{B}} h(B) \tag{UP} \label{eqn:unconstrained-optimization-problem-static-data}
\end{equation}

\subsubsection{The constrained problem}\label{subsubsection:constrained-optimization-problem-static-data}
Problem~\eqref{eqn:unconstrained-optimization-problem-static-data} is a relaxed version of the constrained problem below:
\begin{equation}
\min_{B \in \mathcal{B}} \mathbb{E} d_B(\underline{X},\hat{\underline{X}}) \textbf{ s.t. } \mathbb{E} ||B||_1 \leq \bar{N} \tag{CP} \label{eqn:constrained-optimization-problem-static-data}
\end{equation}
Here the expectation in the constraint is over any possible randomization in choosing the configuration $B$. The cost of activating a sensor, $\lambda$, can be viewed as a Lagrange multiplier used to relax this constrained problem. 

Theorem~\ref{theorem:relation-between-constrained-and-unconstrained-problems} relates solution of \eqref{eqn:unconstrained-optimization-problem-static-data} to \eqref{eqn:constrained-optimization-problem-static-data}.
\begin{theorem}\label{theorem:relation-between-constrained-and-unconstrained-problems}
Consider problems \eqref{eqn:constrained-optimization-problem-static-data} and  \eqref{eqn:unconstrained-optimization-problem-static-data}. If there exists a Lagrange multiplier $\lambda^* \geq 0$ and a $B^* \in \mathcal{B}$, such that an optimal configuration for \eqref{eqn:unconstrained-optimization-problem-static-data} under  $\lambda=\lambda^*$ is $B^*$, and the constraint in \eqref{eqn:constrained-optimization-problem-static-data} is satisfied with equality under the pair 
$(B^*,\lambda^*)$, then $B^*$ is an optimal configuration for \eqref{eqn:constrained-optimization-problem-static-data}.

In case there exist multiple configurations $B_1^*,B_2^*, \cdots, B_m^*$, a multiplier  $\lambda^* \geq 0$, and a probability mass function $(p_1,p_2,\cdots,p_m)$  such that (i) each of $B_1^*,B_2^*,\cdots, B_m^*$ is optimal for problem~\eqref{eqn:unconstrained-optimization-problem-static-data} under $\lambda^*$, and (ii) $\sum_{i=1}^m p_i ||B_i^*||_1=\bar{N}$,  then an optimal solution for \eqref{eqn:constrained-optimization-problem-static-data} is to choose one configuration from $B_1^*,B_2^*,\cdots, B_m^*$ with probability mass function  $(p_1,p_2,\cdots,p_m)$.
\end{theorem}
\begin{proof}
See Appendix~\ref{appendix:proof-of-relation-between-constrained-and-unconstrained-problems}.
\end{proof}
\begin{remark}
Theorem~\ref{theorem:relation-between-constrained-and-unconstrained-problems} allows us to obtain a solution for \eqref{eqn:constrained-optimization-problem-static-data} from the solution of $\eqref{eqn:unconstrained-optimization-problem-static-data}$ by choosing an appropriate $\lambda^*$; this will be elaborated upon in Section~\ref{section:gibbs-stochastic-approximation-unconstrained-problem}.
\end{remark}

\section{Gibbs sampling approach to solve the unconstrained problem}
\label{section:gibbs-sampling-unconstrained-problem}
In this section, we will provide algorithms based on Gibbs sampling to compute the optimal solution 
for \eqref{eqn:unconstrained-optimization-problem-static-data}.

\subsection{Basic Gibbs sampling}\label{subsection:finite-beta-gibbs-sampling-unconstrained-problem}
Let us denote the distribution $\pi_{\beta}(\cdot)$ over $\mathcal{B}$ as follows:
$$\pi_{\beta}(B):=\frac{e^{-\beta h(B)}}{\sum_{B \in \mathcal{B}}e^{-\beta h(B)}}:=\frac{e^{-\beta h(B)}}{Z_{\beta}}$$. 
Motivated by the theory of  statistical physics, we call the parameter $\beta$ the {\em inverse temperature}, and $Z_{\beta}$ the {\em partition function}. Clearly, $\lim_{\beta \uparrow \infty} \sum_{ B \in \arg \min_{A \in \mathcal{B}} h(A) } \pi_{\beta}(B)=1$. Hence, if we 
can choose a configuration $B$ with probability $\pi_{\beta}(B)$ for a large $\beta>0$, we can approximately solve \eqref{eqn:unconstrained-optimization-problem-static-data}.

Computing $Z_{\beta}$ will require $2^N$ addition operations, and hence it is computationally prohibitive for large $N$. As an alternative, we provide an iterative algorithm based on Gibbs sampling, which requires many fewer computations in each iteration. Gibbs sampling runs a discrete-time Markov chain $\{B(t)\}_{t \geq 0}$ whose stationary distribution is $\pi_{\beta}(\cdot)$. 

The BASICGIBBS~algorithm (Algorithm~\ref{algorithm:basic-gibbs-sampling}) simulates the Markov chain $\{B(t)\}_{t \geq 0}$ for any $\beta>0$. The fusion center runs the algorithm to determine the activation set; as such, the fusion center must create a virtual network graph.

\begin{algorithm}[t!]
\hrule
Choose any initial $B(0) \in \{0,1\}^N$. 
At each discrete time instant $t=0,1,2,\cdots$, pick a random sensor $j_t \in \mathcal{N}$ independently and uniformly. For sensor $j_t$, choose $B_{j_t}(t)=1$ with probability 
$p:=\frac{ e^{-\beta h(B_{-j_t}(t-1),1)}}{e^{-\beta h(B_{-j_t}(t-1),1)}+e^{-\beta h(B_{-j_t}(t-1),0)}}$ and choose $B_{j_t}(t)=0$ with probability 
$(1-p)$.  Choose $B_k(t)=B_k(t-1)$ for all $k \neq j_t$.
\hrule
\caption{BASICGIBBS algorithm}
\label{algorithm:basic-gibbs-sampling}
\vspace{2mm}
\end{algorithm}

\begin{theorem}\label{theorem:convergence-basic-Gibbs-sampling}
The Markov chain $\{B(t)\}_{t \geq 0}$ has a stationary distribution $\pi_{\beta}(\cdot)$ under the 
BASICGIBBS~algorithm.
\end{theorem}
\begin{proof}
Follows from the theory  in \cite[Chapter~$7$]{breamud99gibbs-sampling}).
\end{proof}
\begin{remark}
Theorem~\ref{theorem:convergence-basic-Gibbs-sampling} tells us that if the fusion center runs  BASICGIBBS~algorithm and reaches the steady state distribution of the Markov chain $\{B(t)\}_{t \geq 0}$, then the configuration chosen by the algorithm will have distribution $\pi_{\beta}(\cdot)$. For very large  $\beta>0$, if one runs $\{B(t)\}_{t \geq 0}$ for a sufficiently long, finite time $T$, then the terminal state $B_{T}$ will belong to 
$\arg \min_{B \in \mathcal{B}} h(B)$ with high probability. 
\end{remark}

\subsection{The exact solution}\label{subsection:growing-beta-gibbs-sampling-unconstrained-problem}
BASICGIBBS  is operated with a fixed $\beta$; but, in practice, the optimal soultion of the unconstrained problem~\eqref{eqn:unconstrained-optimization-problem-static-data} is obtained with 
$\beta \uparrow \infty$; this is done by updating  $\beta$  at a slower time-scale than the iterates of BASICGIBBS. This new algorithm is called MODIFIEDGIBBS (Algorithm~\ref{algorthm:gibbs-sampling-with-increasing-beta}).

\begin{theorem}\label{theorem:result-on-weak-and-strong-ergodicity}
Under MODIFIEDGIBBS~algorithm, the Markov chain $\{B(t)\}_{t \geq 0}$ is strongly ergodic, and the limiting probability distribution satisfies $\lim_{t \rightarrow \infty} \sum_{A \in \arg \min_{C \in \mathcal{B}} h(C) }\mathbb{P}(B(t)=A)=1$.
\end{theorem}
\begin{proof}
See  Appendix~\ref{appendix:proof-of-weak-and-strong-ergodicity}. 
We have used the notion of weak and strong ergodicity of time-inhomogeneous Markov chains from 
\cite[Chapter~$6$, Section~$8$]{breamud99gibbs-sampling}), which is  provided in Appendix~\ref{appendix:weak-and-strong-ergodicity}. The proof is similar to the proof of one theorem in  \cite{chattopadhyay-etal16gibbsian-caching-arxiv}, but is given here for completeness.
\end{proof}
\begin{remark}
Theorem~\ref{theorem:result-on-weak-and-strong-ergodicity} shows that we can solve \eqref{eqn:unconstrained-optimization-problem-static-data} {\em exactly} if we run MODIFIEDGIBBS~algorithm for infinite time, in contrast with BASICGIBBS~algorithm which   provides an approximate solution.
\end{remark}
\begin{remark}
For  i.i.d. time varying $\{ \underline{X}(t) \}_{t \geq 0}$ with known joint distribution, we can either: (i) find the optimal configuration $B^*$ using MODIFIEDGIBBS  and use $B^*$ for ever, or (ii) run MODIFIEDGIBBS at the same timescale as $t$, and use the running configuration $B(t)$ for sensor activation; both schemes will minimize the time-average expected cost.
\end{remark}

\begin{algorithm}[t!]
\hrule
This algorithm is same as BASICGIBBS~algorithm except that at   time $t$, we use  
$\beta(t):=\beta(0) \log (1+t)$   to compute the update probabilities, where $\beta(0)>0$, $\beta(0) N \Delta<1$, and   $\Delta:=\max_{B \in \mathcal{B}, A \in \mathcal{B}}|h(B)-h(A)|$.
\hrule
\caption{MODIFIEDGIBBS algorithm}
\label{algorthm:gibbs-sampling-with-increasing-beta}
\vspace{2mm}
\end{algorithm}

\subsection{Convergence rate of BASICGIBBS}
Let $\mu_t$ denote the probability distribution of $B(t)$ under BASICGIBBS.  
Let us consider the transition probability matrix $P$ of the Markov chain $\{X(l)\}_{l \geq 0}$ with  $X(l)=B(lN)$, under the BASICGIBBS~algorithm. Let us recall the definition of the Dobrushin's ergodic coefficient $\delta(P)$ from \cite[Chapter~$6$, Section~$7$]{breamud99gibbs-sampling} for the matrix $P$; using a method similar to that of the proof of Theorem~\ref{theorem:result-on-weak-and-strong-ergodicity}, we can show that $\delta(P) \leq ( 1-\frac{ e^{-\beta N \Delta} }{N^N})$. 
  Then, by \cite[Chapter~$6$, Theorem~$7.2$]{breamud99gibbs-sampling}, we can say that under  BASICGIBBS~algorithm, we have $d_V(\mu_{lN},\pi_{\beta}) \leq d_V(\mu_{0},\pi_{\beta}) \bigg( 1-\frac{ e^{-\beta N \Delta} }{N^N}  \bigg)^l$. We can prove similar bounds for any $t=lN+k$, where $0 \leq k \leq N-1$.
  
  Unfortunately, we are not aware of such a bound for MODIFIEDGIBBS.

\begin{remark}
Clearly, under BASICGIBBS~algorithm, the convergence rate decreases as $\beta$ increases. Hence, there is a trade-off between convergence rate and accuracy of the solution in this case. Also, the rate of convergence decreases with $N$. For MODIFIEDGIBBS~algorithm, the convergence rate is expected to decrease with time. 
\end{remark}

\section{Gibbs sampling and stochastic approximation based approach to solve the constrained problem}
\label{section:gibbs-stochastic-approximation-unconstrained-problem}
In Section~\ref{section:gibbs-sampling-unconstrained-problem}, we presented Gibbs sampling based algorithms for  \eqref{eqn:unconstrained-optimization-problem-static-data}. In this section, we provide an algorithm that updates $\lambda$ with time in order to meet the constraint in 
\eqref{eqn:constrained-optimization-problem-static-data} with equality, and thereby solves \eqref{eqn:constrained-optimization-problem-static-data}  (via Theorem~\ref{theorem:relation-between-constrained-and-unconstrained-problems}). 

\begin{lemma}\label{lemma:active-sensors-decreasing-in-lambda}
The optimal mean number of active sensors, $\mathbb{E}|B|_1$, for the unconstrained problem~\eqref{eqn:unconstrained-optimization-problem-static-data}, decreases with $\lambda$. Similarly, the optimal 
error, $\mathbb{E}d_B(\underline{X},\hat{\underline{X}})$, increases with $\lambda$.
\end{lemma}
\begin{proof}
See Appendix~\ref{appendix:proof-of-active-sensors-decreasing-in-lambda}.
\end{proof}
Lemma~\ref{lemma:active-sensors-decreasing-in-lambda} provides an intuition about how to update $\lambda$ in  BASICGIBBS or in MODIFIEDGIBBS in order  to solve \eqref{eqn:constrained-optimization-problem-static-data}.   We  seek to provide one  algorithm which updates $\lambda(t)$ in each iteration, based on the number of active sensors in the previous iteration. In order to maintain the necessary timescale difference between the $\{B(t)\}_{t \geq 0}$ process and the $\lambda(t)$ update process, we use stochastic approximation (\cite{borkar08stochastic-approximation-book}) based update rules for $\lambda(t)$.

\begin{remark}
The optimal mean number of active sensors, $\mathbb{E}||B||_1$, for the unconstrained problem~\eqref{eqn:unconstrained-optimization-problem-static-data} is a decreasing staircase function of $\lambda$, where each point of discontinuity is associated with a change in the optimizer $B^*(\lambda)$.
\end{remark}

The above remark tells us that the optimal solution of the constrained problem~\eqref{eqn:constrained-optimization-problem-static-data} requires us to randomize between two values of $\lambda$ in case the optimal $\lambda^*$ as in 
Theorem~\ref{theorem:relation-between-constrained-and-unconstrained-problems} belongs to the set of such discontinuities. However, this randomization will require us to update a randomization probability at another timescale; having stochastic approximations running in multiple timescales leads to very slow convergence and hence is not a very practical solution for \eqref{eqn:constrained-optimization-problem-static-data}. Hence, instead of using a varying $\beta(t)$, we use a fixed, but large $\beta$ and update $\lambda(t)$ in an iterative fashion using stochastic approximation.

Before proposing the algorithm,  we provide a result analogous to that in Lemma~\ref{lemma:active-sensors-decreasing-in-lambda}.

\begin{algorithm}[t!]
 \hrule
 Choose any initial $B(0) \in \{0,1\}^N$ and $\lambda(0) \geq 0$. 
At each discrete time instant $t=0,1,2,\cdots$, pick a random sensor $j_t \in \mathcal{N}$ independently and uniformly. For sensor $j_t$, choose $B_{j_t}(t)=1$ with probability 
$p:=\frac{ e^{-\beta h_{\lambda(t)}(B_{-j_t}(t-1),1)}}{e^{-\beta h_{\lambda(t)}(B_{-j_t}(t-1),1)}+e^{-\beta h_{\lambda(t)}(B_{-j_t}(t-1),0)}}$ and choose $B_{j_t}(t)=0$ with probability 
$(1-p)$. For $k  \neq j_t$, we choose $B_k(t)=B_k(t-1)$. 

After this operation, before the $(t+1)$ decision instant, update  $\lambda(t)$ at each node as follows. 

$$\lambda(t+1)=[\lambda(t)+a(t) (||B(t-1)||_1-\bar{N})]_b^c$$
The stepsize $\{a(t)\}_{t \geq 1}$ constitutes a positive sequence such that $\sum_{t=1}^{\infty}a(t)=\infty$  and $\sum_{t=1}^{\infty}a^2(t)=\infty$. The nonnegative projection boundaries $b$ and $c$ for the $\lambda(t)$ iterates are such that  $\lambda^* \in (b,c)$ where $\lambda^*$ is defined in Assumption~\ref{assumption:existence-of-optimal-lambda}.
\hrule
\caption{GIBBSLEARNING algorithm}
\label{algorithm:gibbs-learning-algorithm-for-constrained-problem}
\vspace{2mm}
 \end{algorithm}

\begin{lemma}\label{lemma:active-sensors-decreasing-in-lambda-under-basic-gibbs-sampling}
Under BASICGIBBS~algorithm for any given $\beta>0$, the mean number of active sensors 
$\mathbb{E}||B||_1$ is a continuous and decreasing function of $\lambda$.
\end{lemma}
\begin{proof}
See Appendix~\ref{appendix:proof-of-active-sensors-decreasing-in-lambda-under-basic-gibbs-sampling}.
\end{proof}
Let us fix any $\beta>0$. We make the following feasibility assumption for \eqref{eqn:constrained-optimization-problem-static-data}, under the chosen $\beta>0$. 
\begin{assumption}\label{assumption:existence-of-optimal-lambda}
There exists $\lambda^* \geq 0$ such that the constraint in \eqref{eqn:constrained-optimization-problem-static-data} under 
$\lambda^*$ and BASICGIBBS is met with equality.
\end{assumption}
\begin{remark}
By Lemma~\ref{lemma:active-sensors-decreasing-in-lambda-under-basic-gibbs-sampling}, $\mathbb{E}||B||_1$ continuously decreases in $\lambda$. Hence, if $\bar{N}$ is feasible, then such a $\lambda^*$ must exist by the {\em intermediate value theorem}.
\end{remark}
Let us define:
$h_{\lambda(t)}(B):=\mathbb{E}d_B(\underline{X},\hat{\underline{X}})+\lambda(t) ||B||_1$.

Our proposed algorithm GIBBSLEARNING (Algorithm~\ref{algorithm:gibbs-learning-algorithm-for-constrained-problem})  updates $\lambda(t)$ iteratively in order to  solve \eqref{eqn:constrained-optimization-problem-static-data}.

{\em Discussion of GIBBSLEARNING~algorithm:}
\begin{itemize}
\item If $||B(t-1)||_1$ is more than $\bar{N}$, then $\lambda(t)$ is increased with the hope that this will reduce the number of active sensors in subsequent iterations, as suggested by Lemma~\ref{lemma:active-sensors-decreasing-in-lambda-under-basic-gibbs-sampling}.
\item The $B(t)$ and $\lambda(t)$ processes run on two different timescales; $B(t)$ runs in the faster timescale whereas $\lambda(t)$ runs in a slower timescale. This can be understood from the fact that the stepsize in the $\lambda(t)$ update process decreases with time $t$. Here the faster timescale iterate will view the slower timescale iterate as quasi-static, while the slower timescale iterate will view the faster timescale as almost equilibriated. This is reminiscent of  two-timescale stochastic approximation (see \cite[Chapter~$6$]{borkar08stochastic-approximation-book}).
\end{itemize}
    
Let $\pi_{\beta| \lambda^*}(\cdot)$ denote  $\pi_{\beta}(\cdot)$ under $\lambda=\lambda^*$.

\begin{theorem}\label{theorem:optimality-of-the-learning-algorithm-for-constrained-problem}
Under GIBBSLEARNING~algorithm and Assumption~\ref{assumption:existence-of-optimal-lambda}, we have $\lambda(t) \rightarrow \lambda^*$ almost surely, and the limiting distribution of $\{B(t)\}_{t \geq 0}$ is $\pi_{\beta| \lambda^*}(\cdot)$.
\end{theorem}
\begin{proof}
See Appendix~\ref{appendix:proof-of-optimality-of-the-learning-algorithm-for-constrained-problem}.
\end{proof}
This theorem says that GIBBSLEARNING   produces a configuration from the distribution $\pi_{\beta| \lambda^*}(\cdot)$ under steady state.

\subsection{A hard constraint on the number of activated sensors}\label{subsection:hard-constraint}
Let us consider the following modified constrained problem:
\begin{equation}\label{eqn:constrained-optimization-problem-static-data-parametric-distribution}
\min_{B \in \mathcal{B}} \mathbb{E} d_B(\underline{X},\hat{\underline{X}}) \textbf{ s.t. } ||B||_1 \leq \bar{N} 
\tag{MCP}
\end{equation}
It is easy to see that \eqref{eqn:constrained-optimization-problem-static-data-parametric-distribution} can be easily solved using similar Gibbs sampling algorithms as in Section~\ref{section:gibbs-sampling-unconstrained-problem}, where the Gibbs sampling algorithm runs only on the set of configurations which activate $\bar{N}$ number of sensors. 
Thus, as a by-product, we have also proposed a methodology for the problem in \cite{wang-etal16efficient-observation-selection}, though our framework is more general than   \cite{wang-etal16efficient-observation-selection}. 

\begin{remark}
The constraint in \eqref{eqn:constrained-optimization-problem-static-data} is weaker than  \eqref{eqn:constrained-optimization-problem-static-data-parametric-distribution}. 
\end{remark}
\begin{remark}
If we choose $\beta$ very large, then the number of sensors activated by GIBBSLEARNING will have very small variance. This allows us to solve \eqref{eqn:constrained-optimization-problem-static-data-parametric-distribution} with high probability.
\end{remark}

 \begin{algorithm}[t!]
 \hrule
 Choose any initial estimate $\underline{\theta}_1$. Sample the sensor $j_1=\arg \min_{j \in \mathcal{N}}  \mathbb{E} \bigg( d_{B:B_j=1,||B||_1=1}(\underline{X},\hat{\underline{X}}) \bigg|  \underline{\theta}_1 \bigg)$. In general, after sampling nodes $j_1,j_2,\cdots,j_k$ and observing the partial data $\underline{X}_{j_1}=\underline{x}_{j_1}, \cdots, \underline{X}_{j_k}=\underline{x}_{j_k}$, obtain a new estimate $\underline{\theta}_{k+1}$ by {\em completely} running the EM algorithm using the available partial data and starting from the initial estimate  $\underline{\theta}_k$. 
 Once  $\underline{\theta}_{k+1}$ is obtained, sample 
\begin{eqnarray*}
j_{k+1}&=&\arg \min_{j \in \mathcal{N}, j \notin \{ j_1,\cdots,j_k \} }  \mathbb{E}_{B} \\
&& \bigg( d_B(\underline{X},\hat{\underline{X}}) \bigg| \underline{X}_{j_1}=\underline{x}_{j_1}, \cdots, \underline{X}_{j_k}=\underline{x}_{j_k} ;\underline{\theta}_{k+1} \bigg)
\end{eqnarray*} 
where $B$ is such that $B_j=B_{j_1}=\cdots=B_{j_k}=1$, and $||B||_1=k+1$. 
 Continue this process until the $\overline{N}$-th sensor is sampled.
 \hrule
 \caption{EMSTATIC algorithm}
 \label{algorithm:expectation-maximization-greedy-algorithm}
 \end{algorithm}

\section{Expectation maximization based algorithm for parameterized distribution of data}\label{section:EM-for-parametric-distribution}
In previous sections, we assumed that the joint distribution of $\underline{X}$ is completely known to the fusion center. In case this joint distribution is not known but a parametric form $p(\underline{x}|\underline{\theta})$ of the distribution is known with unknown parameter $\underline{\theta}$, selecting all active sensors at once might be highly suboptimal, and a better approach would be to sample sensor nodes sequentially  and refine the estimate of $\underline{\theta}$ using the data collected from a newly sampled sensor. We use standard expectation maximization (EM) algorithm (see \cite[Section~$5.2$]{hajek-lecture-note}) to refine the estimate of $\underline{\theta}$. 
Hence, we present a greedy algorithm~EMSTATIC (Algorithm~\ref{algorithm:expectation-maximization-greedy-algorithm}) to solve \eqref{eqn:constrained-optimization-problem-static-data-parametric-distribution}:

\begin{remark}
This algorithm is  based on heuristics, and it does not have any optimality guarantee because (i) EM algorithm yields a parameter value which corresponds to only a local maximum of the log-likelihood function of the observed data, and (ii) the greedy algorithm to pick the nodes is suboptimal. 

The performance of EMSTATIC~algorithm depends on the initial value 
$\underline{\theta}_1$, since  $\underline{\theta}_1$ will determine  $\{\underline{\theta}_k \}_{k=1,2,\cdots,\bar{N}}$ and the chosen subset of activated sensors.  If $\underline{\theta}_1$ happens to be initialized at a favourable value, then  EMSTATIC~algorithm might even yield the same optimal subset of sensors as in \ref{eqn:constrained-optimization-problem-static-data-parametric-distribution} with $\underline{\theta}$ known apriori. One trivial example for this case would be when $\overline{N}=1$ and we set $\underline{\theta}_1=\underline{\theta}$.
\end{remark}

In case $\underline{X}(t) \sim p(\underline{x}|\underline{\theta})$  varies in time~$t$ in an i.i.d. fashion,  
we can employ the EMSEQUENTIAL~algorithm (Algorithm~\ref{algorithm:expectation-maximization-greedy-algorithm-iid}) to find the optimal subset of sensors at each discrete time slot $t$.

 \begin{remark}
 The performance of EMSEQUENTIAL~algorithm depends on the initial estimate $\underline{\theta}_1$. {\em Also, the maximization operation $B(1)=\arg \min_{B(1) \in \mathcal{B}: ||B(1)||_1=\overline{N}}  \mathbb{E} \bigg( d_{B(1)}(\underline{X}(1),\hat{\underline{X}}(1)) \bigg|  \underline{\theta}_1 \bigg)$ can be efficiently done by employing Gibbs sampling algorithms as in Section~\ref{section:gibbs-sampling-unconstrained-problem}; this shows the potential use of Gibbs sampling in solving sensor subset selection problem for parameterized distribution of data with unknown parameters.} However, since this is not the main focus of our paper, we will only consider known data distribution from now on.
 \end{remark}

\section{Numerical Results}\label{section:numerical-results}

\subsection{Performance of BASICGIBBS~algorithm}
\label{subsection:numerical-performance-of-basic-gibbs-sampling}
For the sake of illustration, we consider $N=18$~sensors which are supposed to sense $\underline{X}=\{X_1,X_2,\cdots,X_{18}\}$, where 
$\underline{X}$ is a jointly Gaussian random vector with covariance matrix $M$. Sensor~$k$ has access only to $X_k$. The matrix $M$ is chosen as follows. We generated a random $N \times N$ matrix $A$ whose elements are uniformly and independently distributed over the interval $[-1,1]$, and set $M=A^T A$ as the covariance matrix of $\underline{X}$. We set sensor activation cost $\lambda=2.3$, and seek to solve \eqref{eqn:unconstrained-optimization-problem-static-data} with MMSE as the error metric. We assume that sensing at each node is perfect,\footnote{However, our analysis can be extended where there is sensing error, but the distribution of sensing error is known to the fusion center.} and 
that the fusion center estimates $\hat{\underline{X}}$ from the observation $\{X_i\}_{i \in S}=:\underline{X}_S$ as $\mathbb{E}(\underline{X} | \underline{X}_S)$, where $S$ is the set of active sensors. Under such an estimation scheme, the conditional distribution of 
$X_{S^c}$ is still a jointly Gaussian random vector with mean $\mathbb{E}(\underline{X}_{S^c} | \underline{X}_S)$ and the covariance matrix $M(S^c,S^c)-M(S^c,S) M(S,S)^{-1}M(S,S^c)$ (see \cite[Proposition~$3.4.4$]{hajek-lecture-note}), where $M(S,S^c)$ is the restriction of $M$ to the rows indexed by $S$ and the columns indexed by $S^c$. The trace of this covariance matrix gives the MMSE when the subset $S$ of sensors are active.

\begin{algorithm}[t!]
 \hrule
 Choose any initial estimate $\underline{\theta}_1$. In slot~$t=1$,  choose the configuration $B(1)$ of sensors  $B(1)=\arg \min_{B(1) \in \mathcal{B}: ||B(1)||_1=\overline{N}}  \mathbb{E} \bigg( d_{B(1)}(\underline{X}(1),\hat{\underline{X}}(1)) \bigg|  \underline{\theta}_1 \bigg)$. Then update the parameter to $\underline{\theta}_2$ using EM algorithm with the partial observation $\underline{X}_{B(1)}(1)=\underline{x}_{B(1)}(1)$ and with initial estimate $\underline{\theta}_1$.   Use $\underline{\theta}_2$ to choose $B(2)=\arg \min_{B(2) \in \mathcal{B}: ||B(2)||_1=\overline{N}}  \mathbb{E} \bigg( d_{B(2)}(\underline{X}(2),\hat{\underline{X}}(2)) \bigg| \underline{\theta}_2 \bigg)$ in slot $t=2$. Continue this procedure for all $t$.
 \hrule
 \caption{EMSEQUENTIAL algorithm}
 \label{algorithm:expectation-maximization-greedy-algorithm-iid}
 \vspace{2mm}
 \end{algorithm}

In this scenario, in Figure~\ref{fig:comparison--gibbs-optimal-greedy-finitegibbs}, we compare the cost for the following four algorithms:
\begin{itemize}
\item {\em OPTIMAL:} Here we consider the minimum possible cost for \eqref{eqn:unconstrained-optimization-problem-static-data}.
\item {\em BASICGIBBS under steady state:} Here we assume that the configuration $B \in \mathcal{B}$ is chosen 
according to the distribution $\pi_{\beta}(\cdot)$ defined in Section~\ref{section:gibbs-sampling-unconstrained-problem}. 
This is done for several values of $\beta$.
\item {\em BASICGIBBS  with finite iteration:} Here we run BASICGIBBS~algorithm for $100$~iterations. This is done independently for several values of $\beta$, where for each $\beta$ the iteration starts from an independent random configuration. Note that, we have simulated only one sample path of BASICGIBBS for each $\beta$; if the algorithm is run again independently, the results will be different.
\item {\em GREEDY:} Start with an empty set $S$, and find the cost if this subset of sensors are activated. Then compare this cost with the cost in case sensor~$1$ is added to this set. If it turns out that adding sensor~$1$ to this set $S$ reduces the cost, then add sensor~$1$ to the set $S$; otherwise, remove sensor~$1$ from set $S$. Do this operation serially for all sensors, and activate the  sensors given by the final set $S$.
\end{itemize}
It turns out that, under the optimal configuration, $12$ sensors are activated and the optimal cost is $32.3647$. On the other hand, GREEDY activates $14$~sensors and incurred a cost of $35.9663$. However, we are not aware of any monotonicity or supermodularity property of the objective function in \eqref{eqn:unconstrained-optimization-problem-static-data}; hence, we cannot provide any constant approximation ratio guarantee for the problem~\eqref{eqn:unconstrained-optimization-problem-static-data}. On the other hand, we have already proved that BASICGIBBS performs near optimally for large $\beta$. Hence, we choose to investigate the performance of BASICGIBBS, though it might require more number of iterations compared to $N=18$ iterations for GREEDY. It is important to note that, \eqref{eqn:unconstrained-optimization-problem-static-data} is NP-hard, and BASICGIBBS  allows us to avoid searching over $2^N$ possible configurations.

In Figure~\ref{fig:comparison--gibbs-optimal-greedy-finitegibbs}, we can see that for $\beta \geq 2$, the steady state distribution $\pi_{\beta}(\cdot)$ of BASICGIBBS achieves better expected cost than GREEDY, and the cost becomes closer to the optimal cost as $\beta$ increases. On the other hand, for each $\beta \geq 2$, BASICGIBBS after $100$~iterations yielded a configuration   that achieves near-optimal cost. Hence, BASICGIBBS with reasonably  small number of iterations can be used to find the optimal subset of active sensors when $N$ is large.

 \begin{figure}[!t]
\begin{center}
\includegraphics[height=7cm, width=9cm]{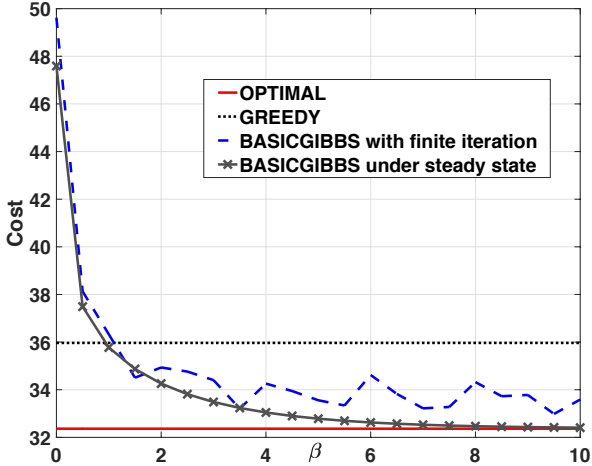}
\end{center}
\caption{Comparison among OPTIMAL, BASICGIBBS under steady state, GREEDY, and BASICGIBBS with finite iterations, for solving problem~\eqref{eqn:unconstrained-optimization-problem-static-data}.   For each $\beta$, BASICGIBBS with finite iterations stops after $100$~iterations. Details are provided in Section~\ref{subsection:numerical-performance-of-basic-gibbs-sampling}.}
\label{fig:comparison--gibbs-optimal-greedy-finitegibbs}
\end{figure}

\subsection{Performance of Gibbs sampling applied to problem~\eqref{eqn:constrained-optimization-problem-static-data-parametric-distribution}}\label{subsection:numerical-gibbs-sampling-applied-to-hard-constrained-problem}
Here we seek to solve problem~\eqref{eqn:constrained-optimization-problem-static-data-parametric-distribution} with $\bar{N}=10$ under the same setting as in Section~\ref{subsection:numerical-performance-of-basic-gibbs-sampling} except that a new sample of  the covariance matrix $M$ is chosen. Here we compare the estimation error for the following three cases:
\begin{itemize}
\item {\em OPTIMAL:} Here we choose an optimal subset  for \eqref{eqn:constrained-optimization-problem-static-data-parametric-distribution}.
\item {\em BASICGIBBS  under steady state:} Here we assume that the configuration $B$ is chosen 
according to the steady-state distribution $\pi_{\beta}(\cdot)$ defined in Section~\ref{section:gibbs-sampling-unconstrained-problem}, but restricted only to the set $\{B \in \mathcal{B}: ||B||_1=\bar{N}\}$. This is done by putting $h(B)=\mathbb{E}d_B(\underline{X},\hat{\underline{X}})$ if $||B||_1=\bar{N}$ and $h(B)=\infty$ otherwise.
This is done for several values of $\beta$.
\item {\em NEWGREEDY:} Start with an empty set $S$, and find the estimation error if this subset of sensors are activated. Then find the sensor~$j_1$ which, when added to $S$, will result in the minimum estimation error. If the estimation error for $S \cup \{j_1\}$ is less than that of $S$, then do $S=S \cup \{j_1\}$.  Now find the sensor~$j_2$ which, when added to $S$, will result in the minimum estimation error. If the estimation error for $S \cup \{j_2\}$ is less  than that of $S$, then do $S=S \cup \{j_2\}$.  Repeat this operation until we have $|S|=\bar{N}$, and activate the set of $\bar{N}$ sensors given by the final set $S$. 
{\em A similar greedy algorithm is used in \cite{wang-etal16efficient-observation-selection}}.
\end{itemize}
The performances for these three cases are shown in Figure~\ref{fig:gibbs_optimal_greedy_comparison_fixed_number_of_active_sensors}. It turns out that, the estimation error for  OPTIMAL and NEWGREEDY are $12.9741$ and $15.4343$, respectively. BASICGIBBS outperforms NEWGREEDY for $\beta \geq 2$, and becomes very close to OPTIMAL performance for $\beta \geq 5$.

\begin{figure}[!t]
\begin{center}
\includegraphics[height=7cm, width=9cm]{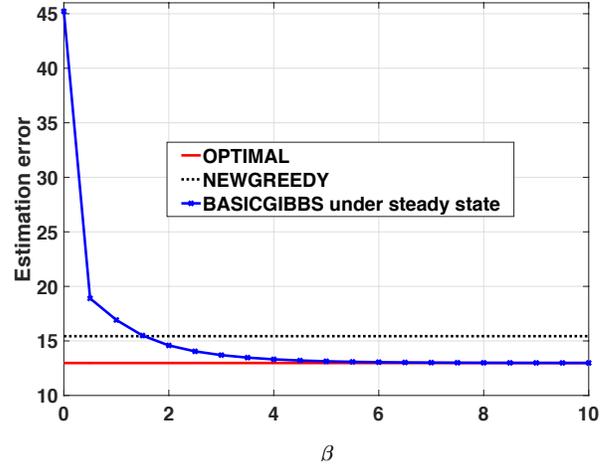}
\end{center}
\vspace{-6mm}
\caption{Comparison among OPTIMAL, BASICGIBBS under steady state, and NEWGREEDY, for solving problem~\eqref{eqn:constrained-optimization-problem-static-data-parametric-distribution}.   Details are provided in  Section~\ref{subsection:numerical-gibbs-sampling-applied-to-hard-constrained-problem}.}
\label{fig:gibbs_optimal_greedy_comparison_fixed_number_of_active_sensors}
\end{figure}

\begin{figure}[!t]
\begin{center}
\includegraphics[height=7cm, width=8cm]{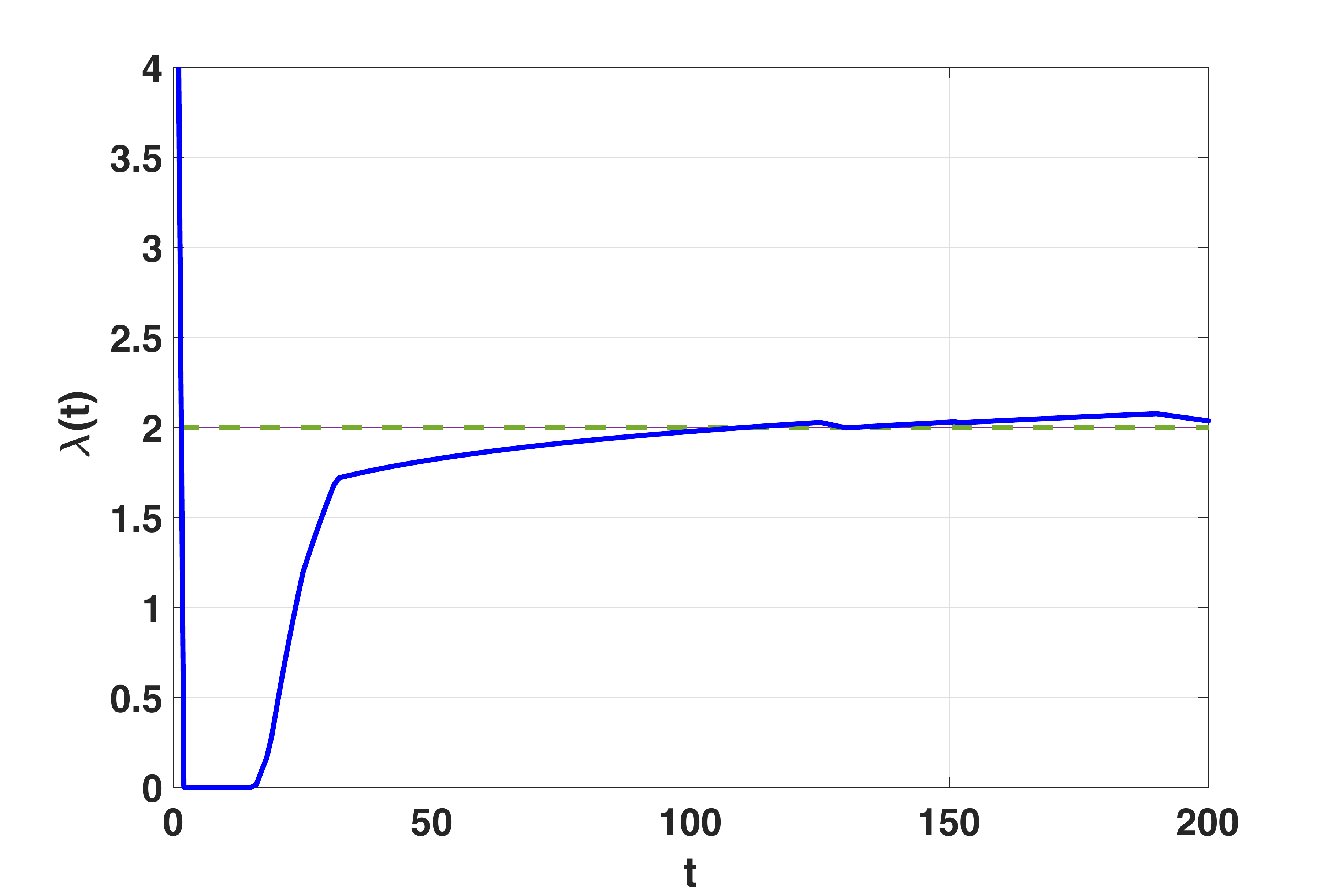}
\vspace{-6mm}
\includegraphics[height=7cm, width=8cm]{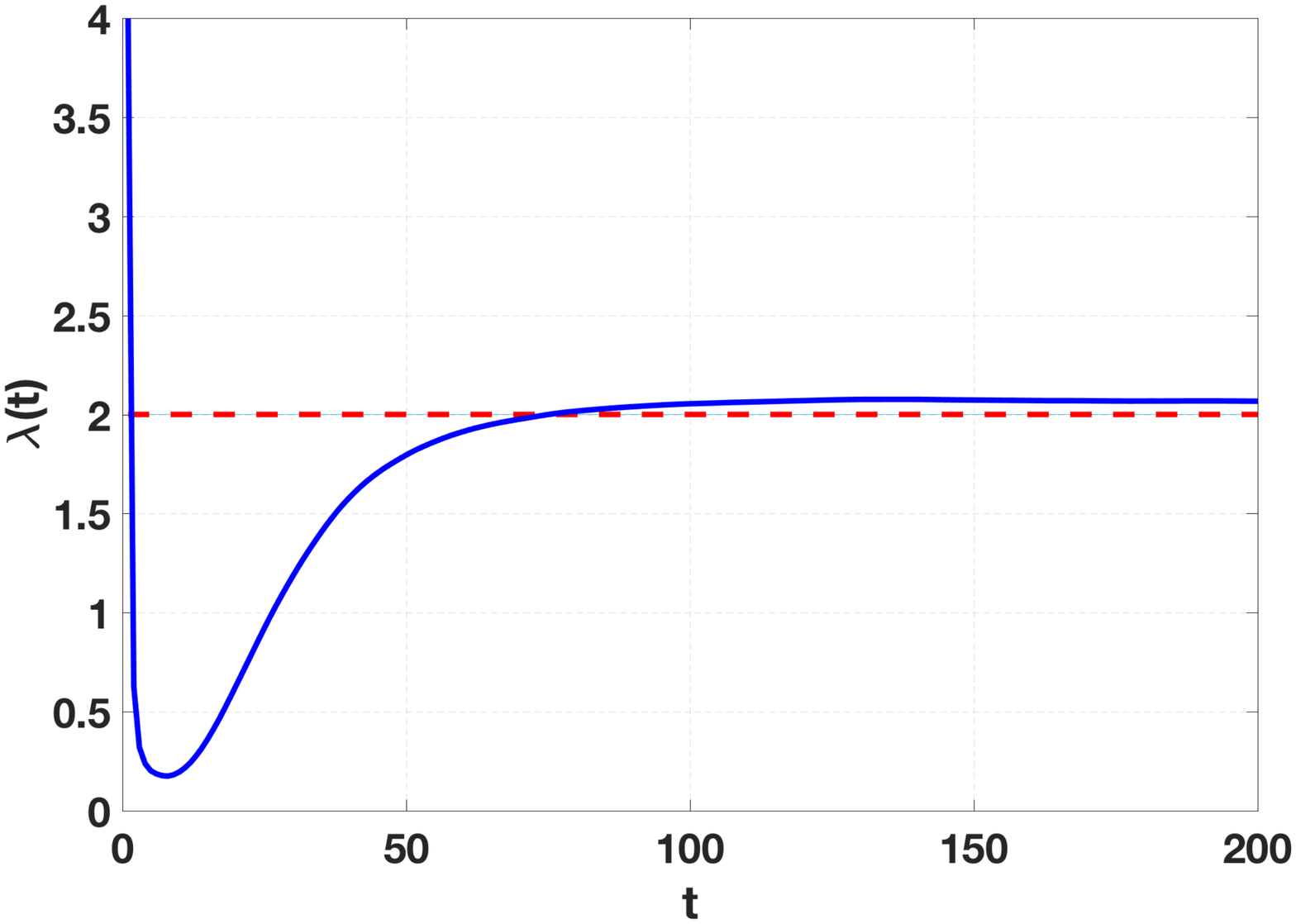}
\end{center}
\vspace{-4mm}
\caption{Illustration for convergence speed of $\lambda(t)$ in the GIBBSLEARNING~algorithm. 
{\bf Top plot:} Result for a single sample path. Details can be found in   Section~\ref{subsection:numerical-convergence-speed-gibbs-sampling-stochastic-approximation}. 
{\bf Bottom plot:} Average result over $1000$ independent sample paths. Details can be found in   Section~\ref{subsection:numerical-convergence-speed-gibbs-sampling-stochastic-approximation}.}
\label{fig:gibbs-stochastic-approximation}
\end{figure}

\subsection{Convergence speed of GIBBSLEARNING~algorithm}
\label{subsection:numerical-convergence-speed-gibbs-sampling-stochastic-approximation}
We first demonstrate the convergence speed of GIBBSLEARNING~algorithm, {\em for one specific sample path}.

We consider a setting similar to that of Section~\ref{subsection:numerical-performance-of-basic-gibbs-sampling}, except that we  fix $\beta=5$. The covariance matrix $M$ is generated using the same method, but the realization of $M$ here is different from that in Section~\ref{subsection:numerical-performance-of-basic-gibbs-sampling}. Under this setting, for $\lambda^*=2$, BASICGIBBS~algorithm yields the MMSE $3.5680$, and the expected number of sensors activated by BASICGIBBS~algorithm becomes $12.7758$. Now, let us consider problem~\eqref{eqn:constrained-optimization-problem-static-data} with the constraint value $\bar{N}=12.7758$. Clearly, if GIBBSLEARNING~algorithm is employed to find out the solution of problem~\eqref{eqn:constrained-optimization-problem-static-data} with $\bar{N}=12.7758$, then $\lambda(t)$ should converge to $\lambda^*=2$. 

The evolution of $\lambda(t)$ against the iteration index $t$ is shown in the top plot in Figure~\ref{fig:gibbs-stochastic-approximation}. We can see that, starting from $\lambda(0)=4$ and and using the stepsize sequence $a(t)=\frac{1}{t}$, the iterate  $\lambda(t)$ becomes very close to $\lambda^*=2$ within $100$~iterations.  At $t=200$, we found that $\lambda(200)=2.0318$. The resulting configuration yielded by GIBBSLEARNING~algorithm at $t=200$ achieves MMSE $5.0308$ and activates $12$~sensors; it is important to remember that these results are for one specific realization of the sample path. On the other hand, under $\lambda(200)=2.0318$, the steady state distribution of BASICGIBBS, $\pi_{\beta}(\cdot)$, yields MMSE $3.6524$ and mean number of active sensors $12.7354$, which are very close to the respective target values $3.5680$ and $\bar{N}=12.7758$. 

However, the top plot in Figure~\ref{fig:gibbs-stochastic-approximation} is only for a specific sample path of GIBBSLEARNING~algorithm. In the bottom plot in Figure~\ref{fig:gibbs-stochastic-approximation}, we demonstrate convergence speed of $\lambda(t)$ averaged over multiple independent sample paths of GIBBSLEARNING~algorithm. Here we generate a {\em different covariance matrix} $M$, set $\lambda^*=2$, and follow the same procedure as before to set $\bar{N}$. Then we run GIBBSLEARNING~algorithm independently $1000$~times, starting from  $\lambda(0)=4$. The bottom plot of Figure~\ref{fig:gibbs-stochastic-approximation} shows the variation of $\lambda(t)$ (averaged over $1000$ sample paths) with $t$. We can again see that the average $\lambda(t)$ is very close to $\lambda^*=2$ for $t \geq 100$.

Thus, our numerical illustration shows that GIBBSLEARNING~algorithm has reasonably fast convergence rate  for practical active sensing.

\section{Conclusion}\label{section:conclusion}
In this paper, we have presented Gibbs sampling, stochastic approximation and expectation maximization based algorithms for efficient data estimation in the context of active sensing. We first proposed Gibbs sampling based algorithms for unconstrained optimization of the estimation error and the mean number of active sensors, proved convergence of these algorithms, and provided a bound on the convergence speed. Next, we proposed an algorithm based on Gibbs sampling and stochastic approximation, in order to solve a constrained version of the above unconstrained problem, and proved its convergence. Finally, we proposed expectation maximization based algorithms for the scenario where the sensor data is coming from a distribution with known parametric distribution but unknown parameter value.   Numerical results demonstrate the near-optimal performance of some of these algorithms with small number of computations. 

As our future research endeavours, we seek to develop distributed sensor subset selection algorithms to efficiently track the data varying in time according to a stochastic process.

{\small
\bibliographystyle{unsrt}
\bibliography{arpan-techreport}
}

\newpage

\appendices

\section{Proof of Theorem~\ref{theorem:relation-between-constrained-and-unconstrained-problems}}
\label{appendix:proof-of-relation-between-constrained-and-unconstrained-problems}
We will prove only the first part of the theorem where there exists a unique $B^*$. The second part of the theorem can be proved similarly. 

Let us denote the optimizer for \eqref{eqn:constrained-optimization-problem-static-data} by $B$, which is possibly different from $B^*$. Then, by the definition of $B^*$, we have $\mathbb{E} d_{B*}(\underline{X},\hat{\underline{X}})+ \lambda^* ||B^*||_1 \leq \mathbb{E} d_B (\underline{X},\hat{\underline{X}}) + \lambda^* ||B||_1$. But $||B||_1 \leq K$ and $||B^*||_1=K$. Hence, $\mathbb{E} d_{B*}(\underline{X},\hat{\underline{X}}) \leq \mathbb{E} d_B (\underline{X},\hat{\underline{X}})$. This completes the proof.

\section{Weak and Strong Ergodicity}
\label{appendix:weak-and-strong-ergodicity}
Consider a discrete-time  Markov chain (possibly not time-homogeneous) $\{B(t)\}_{t \geq 0}$ with transition probability matrix (t.p.m.)  $P(m;n)$ between 
$t=m$ and $t=n$. We denote by $\mathcal{D}$   the collection of all possible probasbility distributions 
 on the state space. Let $d_V(\cdot,\cdot)$ denote  the total variation distance between two distributions in $\mathcal{D}$. 
Then $\{B(t)\}_{t \geq 0}$ is called weakly ergodic if, for all $m \geq 0$,  we have 
$\lim_{n \uparrow \infty} \sup_{\mu,\nu \in \mathcal{D}}  d_V (\mu P(m;n) , \nu P(m;n) ) =0 $.

The Markov chain $\{B(t)\}_{t \geq 0}$ is called strongly ergodic if there exists $\pi \in \mathcal{D}$ such that, 
$\lim_{n \uparrow \infty} \sup_{\mu \in \mathcal{D}}  d_V (\mu^{T} P(m;n) , \pi ) =0 $ for all $m \geq 0$.

\section{Proof of Theorem~\ref{theorem:result-on-weak-and-strong-ergodicity}}
\label{appendix:proof-of-weak-and-strong-ergodicity}
 
We will first show that the Markov chain $\{B(t)\}_{t \geq 0}$ in weakly ergodic.

Let us define $\Delta:=\max_{B \in \mathcal{B}, A \in \mathcal{B}}|h(B)-h(A)|$.

Consider the transition probability matrix (t.p.m.) $P_l$ for the inhomogeneous Markov 
chain $\{X(l)\}_{l \geq 0}$ (where $X(l):=B(lN)$). The Dobrushin's ergodic coefficient $\delta(P_l)$ is given by 
(see \cite[Chapter~$6$, Section~$7$]{breamud99gibbs-sampling} for definition) 
$\delta(P_l)=1- \inf_{B^{'},B^{''} \in \mathcal{B}} \sum_{B \in \mathcal{B}} \min \{P_l(B^{'},B),P_l(B^{''},B) \}$. 
A sufficient condition for the Markov chain $\{B(t)\}_{t \geq 0}$  
to be weakly ergodic is $\sum_{l=1}^{\infty}(1-\delta(P_l))=\infty$ (by 
\cite[Chapter~$6$, Theorem~$8.2$]{breamud99gibbs-sampling}).

Now, with positive probability, activation states for all nodes are updated over a period of $N$ slots. Hence, $P_l(B^{'},B)>0$ for all $B^{'},B \in \mathcal{B}$. Also, once a node $j_t$ for $t=lN+k$ 
is chosen  in MODIFIEDGIBBS~algorithm, the sampling probability for any activation state in a slot 
is greater than  $\frac{e^{-\beta(lN+k) \Delta}}{2}$. 
Hence, for independent sampling over $N$ slots, we  have, for all pairs $B^{'},B$: 
$$P_l(B^{'},B) >  \prod_{k=0}^{N-1}\bigg( \frac{e^{-\beta(lN+k) \Delta}}{2N} \bigg) >0$$ 
Hence, 
\begin{eqnarray}
&&\sum_{l=0}^{\infty}(1-\delta(P_l)) \nonumber\\
&=& \sum_{l=0}^{\infty} \inf_{B^{'},B^{''} \in \mathcal{B}} \sum_{B \in \mathcal{B}} \min \{P_l(B^{'},B),P_l(B^{''},B) \} \nonumber\\
& \geq & \sum_{l=0}^{\infty} 2^N \prod_{k=0}^{N-1} \bigg( \frac{e^{-\beta(0) \log(1+lN+k) \times \Delta}}{2N} \bigg) \nonumber\\
& \geq &   \sum_{l=0}^{\infty}  \prod_{k=0}^{N-1} \bigg( \frac{e^{-\beta(0) \log(1+lN+N) \times \Delta}}{N} \bigg)  \nonumber\\
& = &  \frac{1}{N^N} \sum_{l=1}^{\infty}  \frac{1}{  (1+lN)^{\beta(0) N \Delta}}  \nonumber\\
& \geq &  \frac{1}{ N^{N+1}} \sum_{i=N+1}^{\infty}  \frac{1}{  (1+i)^{\beta(0) N \Delta}}  \nonumber\\
& = & \infty
 \end{eqnarray}
Here the first inequality uses the fact that the cardinality of $\mathcal{B}$ is $2^N$. The second inequality follows from  replacing $k$ by $N$ in the numerator. The third inequality follows from lower-bounding  
$\frac{1}{(1+lN)^{\beta(0) N \Delta}}$ by $\frac{1}{N}\sum_{i=lN}^{lN+N-1}  \frac{1}{  (1+i)^{\beta(0) N \Delta}} $. 
The last equality follows from the fact that $\sum_{i=1}^{\infty} \frac{1}{i^a}$ diverges for $0 <a<1$.

 Hence, the Markov chain  $\{B(t)\}_{t \geq 0}$ is  weakly ergodic.
 
 In order to prove strong ergodicity of 
 $\{B(t)\}_{t \geq 0}$, we invoke \cite[Chapter~$6$, Theorem~$8.3$]{breamud99gibbs-sampling}. 
 We denote the  t.p.m. of $\{B(t)\}_{t \geq 0}$ at a specific time $t=T_0$ 
 by $Q^{(T_0)}$, which is a given specific matrix. If  $\{B(t)\}_{t \geq 0}$   evolves up to infinite time 
 with {\em fixed} t.p.m. $Q^{(T_0)}$, then it will reach the stationary distribution $\pi_{\beta_{T_0}}(B)= \frac{e^{-\beta_{T_0} h(B)}}{Z_{\beta_{T_0}}}$. 
 Hence, we can claim that Condition~$8.9$ of \cite[Chapter~$6$, Theorem~$8.3$]{breamud99gibbs-sampling} is  satisfied. 
 
 Next, we check Condition~$8.10$ of \cite[Chapter~$6$, Theorem~$8.3$]{breamud99gibbs-sampling}. 
 For any $B \in \arg \min_{B^{'} \in \mathcal{B}} h(B^{'})$,  we can argue that $\pi_{\beta_{T_0}}(B)$ increases with $T_0$ for 
 sufficiently large $T_0$; this can be verified by considering the derivative of $\pi_{\beta}(B)$ w.r.t. $\beta$. For $B \notin \arg \min_{B^{'} \in \mathcal{B}} h(B^{'})$,  the probability 
 $\pi_{\beta_{T_0}}(B)$ decreases with $T_0$ for large $T_0$. Now, using the fact that any monotone, bounded sequence converges, we can write $\sum_{T_0=0}^{\infty} \sum_{B \in \mathcal{B}} |\pi_{\beta_{T_0+1}}(B)-\pi_{\beta_{T_0}}(B)| < \infty$. 
 
 Hence, by \cite[Chapter~$6$, Theorem~$8.3$]{breamud99gibbs-sampling}, the Markov chain $\{B(t)\}_{t \geq 0}$ is strongly ergodic.  It  is straightforward to verify the claim regarding the  limiting distribution.

\section{Proof of Lemma~\ref{lemma:active-sensors-decreasing-in-lambda}}
\label{appendix:proof-of-active-sensors-decreasing-in-lambda}
Let $\lambda_1 > \lambda_2 > 0$, and the corresponding optimal error and mean number of active sensors under these multiplier values be $(d_1,n_1)$ and $(d_2,n_2)$, respectively. Then, by definition, $d_1 + \lambda_1 n_1 \leq d_2 + \lambda_1 n_2$ and $d_2 + \lambda_2 n_2 \leq d_1 + \lambda_2 n_1$. Adding these two inequalities, we obtain $\lambda_1 n_1+\lambda_2 n_2 \leq \lambda_1 n_2+\lambda_2 n_1$, i.e., $(\lambda_1-\lambda_2)n_1 \leq (\lambda_1-\lambda_2)n_2$. Since $\lambda_1>\lambda_2$, we obtain $n_1 \leq n_2$. This completes the first part of the proof. The second part of the proof follows using similar arguments.

\section{Proof of Lemma~\ref{lemma:active-sensors-decreasing-in-lambda-under-basic-gibbs-sampling}}
\label{appendix:proof-of-active-sensors-decreasing-in-lambda-under-basic-gibbs-sampling}
Let us denote $\mathbb{E}||B||_1=:f(\lambda)=\frac{\sum_{B \in \mathcal{B}} ||B||_1 e^{-\beta h(B)}}{Z_{\beta}}$. 
It is straightforward to see that $\mathbb{E}||B||_1$ is continuously differentiable in $\lambda$. 
 
Let us denote $Z_{\beta}$ by $Z$ for simplicity, and let $h(B)=d_B+\lambda ||B||_1$ be the linear combination of the error (here we have written $\mathbb{E}d_B(\cdot,\cdot)$ as $d_B$ for simplicity in notation) and number of active sensors under configuration $B$. 
Then the derivative of $f(\lambda)$ w.r.t. $\lambda$ is given by:
\small
\begin{eqnarray*}
&& f'(\lambda)\\
&=& \frac{  -Z \beta \sum_{B \in \mathcal{B}} ||B||_1^2 e^{-\beta(d_B+\lambda ||B||_1)} - \sum_{B \in \mathcal{B}} ||B||_1 e^{-\beta(d_B+\lambda ||B||_1)}  \frac{dZ}{d \lambda}        }{Z^2}
\end{eqnarray*}
\normalsize

Now, it is straightforward to verify that $\frac{dZ}{d \lambda}=-\beta Z f(\lambda)$. Hence, 
\tiny
\begin{eqnarray*}
&& f'(\lambda)\\
&=& \frac{  -Z \beta \sum_{B \in \mathcal{B}} ||B||_1^2 e^{-\beta(d_B+\lambda ||B||_1)} + \sum_{B \in \mathcal{B}} ||B||_1 e^{-\beta(d_B+\lambda ||B||_1)} \beta Z f(\lambda)       }{Z^2}
\end{eqnarray*}
\normalsize

Now, $f'(\lambda) \leq 0$ is equivalent to 
$$f(\lambda) \leq \frac{  \sum_{B \in \mathcal{B}} ||B||_1^2 e^{-\beta(d_B+\lambda ||B||_1)}    }{    \sum_{B \in \mathcal{B}} ||B||_1 e^{-\beta(d_B+\lambda ||B||_1)}    }$$
Noting that $\mathbb{E}||B||_1=:f(\lambda)$ and dividing the numerator and denominator of R.H.S. by $Z$, 
the condition is reduced to $\mathbb{E}||B||_1 \leq \frac{\mathbb{E}||B||_1^2}{\mathbb{E}||B||_1}$, which is true since 
$\mathbb{E}||B||_1^2 \geq (\mathbb{E}||B||_1)^2$. Hence, $\mathbb{E}||B||_1$ is decreasing in $\lambda$ for any $\beta>0$.

\section{Proof of Theorem~\ref{theorem:optimality-of-the-learning-algorithm-for-constrained-problem}}
\label{appendix:proof-of-optimality-of-the-learning-algorithm-for-constrained-problem}
Let the distribution of $B(t)$ under GIBBSLEARNING~algorithm be 
$\mu_t(\cdot)$. Since $\lim_{t \rightarrow \infty} a(t)=0$, it follows that $\lim_{t \rightarrow \infty} d_V(\mu_{t-1}, \pi_{\beta| \lambda(t-1)})=0$ (where $d_V(\cdot,\cdot)$ is the total variation distance), and $\lim_{t \rightarrow \infty} ( \mathbb{E}_{\mu_{t-1}}||B||_1-\mathbb{E}_{\pi_{\beta} | \lambda(t-1) }||B||_1 ) :=\lim_{t \rightarrow \infty} e(t)=0$. Now, we can rewrite the $\lambda(t)$ update equation as follows:
\small
\begin{eqnarray}
 \lambda(t+1)=[  \lambda(t)+a(t) (\mathbb{E}_{\pi_{\beta}|\lambda(t-1)}||B||_1-\bar{N}+M_t+e_t)  ]_b^c
\end{eqnarray} 
\normalsize

Here $M_t:=||B(t-1)||_1 - \mathbb{E}_{\mu_{t-1}} ||B(t-1)||_1$ is a Martingale difference noise sequence, and 
$\lim_{t \rightarrow \infty} e_t=0$. It is easy to see that the derivative of $\mathbb{E}_{\pi_{\beta} | \lambda }||B||_1$ w.r.t. $\lambda$ is bouned for $\lambda \in [b,c]$; hence, $\mathbb{E}_{\pi_{\beta} | \lambda }||B||_1$ is a Lipschitz continuous function of $\lambda$. It is also easy to see that the sequence $\{M_t\}_{t \geq 0}$ is bounded. Hence, by the theory presented in 
\cite[Chapter~$2$]{borkar08stochastic-approximation-book} and 
\cite[Chapter~$5$, Section~$5.4$]{borkar08stochastic-approximation-book}, $\lambda(t)$ converges to the unique zero of $\mathbb{E}_{\pi_{\beta} | \lambda }||B||_1-\bar{N}$ almost surely. Hence,  $\lambda(t) \rightarrow \lambda^*$ almost surely. Since $\lim_{t \rightarrow \infty} d_V(\mu_{t-1}, \pi_{\beta| \lambda(t-1)})=0$ and $\pi_{\beta | \lambda}$ is continuous in $\lambda$, the limiting distribution of $B(t)$ becomes $\pi_{\beta | \lambda^*}$.

\end{document}